\documentclass[letterpaper, 10 pt, conference]{ieeeconf}  
\IEEEoverridecommandlockouts                             
\usepackage[dvipsnames]{xcolor}

\usepackage{amssymb} 
\usepackage{amsmath}
\usepackage{amsthm}
\usepackage{xcolor}
\theoremstyle{plain}
\newtheorem{theorem}{Theorem}
\newtheorem{lemma}[theorem]{Lemma}
\newtheorem{remark}{Remark}

\usepackage{graphicx}

\usepackage{cite}

\usepackage[dvipsnames]{xcolor}
\newcommand{\remove}[1]{}

\usepackage{subcaption}
\usepackage{hyperref}
\usepackage{etoolbox}

\AtBeginEnvironment{equation}{\footnotesize} 
\AtBeginEnvironment{align}{\footnotesize} 
\usepackage[font=footnotesize,labelfont=bf]{caption}

\usepackage{soul} 
\title{\LARGE \bf
Adaptive Deadlock Avoidance for Decentralized Multi-agent Systems via CBF-inspired Risk Measurement
}

\author{Yanze Zhang$^{1}$, Yiwei Lyu$^{2}$, Siwon Jo$^{3}$, Yupeng Yang$^{3}$, and Wenhao Luo$^{1}$
\thanks{$^*$This work was supported in part by the U.S. National Science Foundation under Grant CNS-2312465.}
\thanks{$^{1}$The authors are with the Department of Computer Science, University of Illinois Chicago, Chicago, IL 60607, USA. Email: {\tt\small \{yzhan361, wenhao\}@uic.edu}}
\thanks{$^{2}$The author is with the Department of Electrical and Computer Engineering, Carnegie Mellon University,
        Pittsburgh, PA 15213, USA. Email:
        {\tt\small yiweilyu@andrew.cmu.edu}}
\thanks{$^{3}$The authors are with the Department of Computer Science, University of North Carolina at Charlotte, Charlotte, NC 28223, USA. Email:
        {\tt\small \{sjo2, yyang52\}@charlotte.edu}}
}

\begin{document}
\maketitle
\thispagestyle{empty}
\pagestyle{empty}

\begin{abstract}
Decentralized safe control plays an important role in multi-agent systems given the scalability and robustness without reliance on a central authority. However, without an explicit global coordinator, the decentralized control methods are often prone to deadlock ---  a state where the system reaches equilibrium, causing the robots to stall. In this paper, we propose a generalized decentralized framework that unifies the Control Lyapunov Function (CLF) and Control Barrier Function (CBF) to facilitate efficient task execution and ensure deadlock-free trajectories for the multi-agent systems. As the agents approach the deadlock-related undesirable equilibrium, the framework can detect the equilibrium and drive agents away before that happens. This is achieved by a secondary deadlock resolution design with an auxiliary CBF to prevent the multi-agent systems from converging to the undesirable equilibrium. To avoid dominating effects due to the deadlock resolution over the original task-related controllers, a deadlock indicator function using CBF-inspired risk measurement is proposed and encoded in the unified framework for the agents to adaptively determine when to activate the deadlock resolution. This allows the agents to follow their original control tasks and seamlessly unlock or deactivate deadlock resolution as necessary, effectively improving task efficiency. We demonstrate the effectiveness of the proposed method through theoretical analysis, numerical simulations, and real-world experiments.
\end{abstract}

\section{Introduction}
With the desirable properties including efficiency and scalability, the multi-agent systems are well-suited for a wide range of applications across various scenarios, e.g., sensor coverage \cite{luo2018adaptive}, environment exploration \cite{banfi2016asynchronous}, warehouse management \cite{bolu2021adaptive}, and search and rescue \cite{baxter2007multi}. 
Decentralized control methods where each agent independently makes decisions by taking into account its local information and interactions with its neighboring agents have gained increasing attention recently. However, without a global coordinator, the decentralized control policy may inevitably introduce the undesired equilibria, where the agents would 
stop before reaching their goals and remain stalled unless intervention occurs,
often referred to as \emph{deadlock} \cite{grover2023before, wang2017safety, celi2019deconfliction}.

To address the issue of deadlock, some
existing approaches have attempted to introduce minor perturbations into the system state or control inputs \cite{zhu2023safety, liu2023realtime, renault2024multi, martinez2024avocado}
before or when
the deadlock happens. However, these methods can not provide a definitive assurance against the ocurrence/recurrence
of deadlock \cite{grover2023before}.
A right-hand heuristic rule is proposed in \cite{celi2019deconfliction}, where
the quasi-deadlock condition is present to detect deadlock. However, the design of the if-else conditions would decrease the system's efficiency, whereas a framework around the original task-related controllers that can detect and fix the deadlock in a minimally invasive manner may be more desirable. 

Other methods analyze the deadlock conditions by utilizing the Karush-Kuhn-Tucker (KKT) conditions for the multi-agent system. The work in
\cite{grover2023before, grover2021deadlock} characterize the deadlock as a force equilibrium and introduce a Proportional-Derivative (PD) control law as a mechanism for deadlock resolution. However, it is noteworthy that the theoretical validity of this approach is confined to scenarios involving no more than three agents. Using the similar idea of force equilibrium, in \cite{chen2022deadlock} an
infinite-horizon model predictive control with deadlock resolution is proposed to prevent deadlocks before they occur.
Additionally,
\cite{reis2020control} summarizes deadlock conditions with detailed proof of the equilibrium points' existence and stability, and subsequently proposes a formal framework to resolve the deadlock. 
\cite{cavorsi2023multi} provides a new representation for the free space where deadlock-prone areas are blocked off by analyzing the KKT conditions.
However, the purpose of avoiding necessary conditions to preclude undesired deadlock regardless of when that may occur could dominate the robots' behaviors from the starting stage, and thus unnecessarily disturb their original task execution.
\cite{desai2023auxiliary} extended this method by integrating an auxiliary control effort into a novel Tracking-CBF-QP framework, which aids in avoiding deadlocks and improves trajectory tracking performance. Nevertheless, it remains challenging 
for individual agents to determine when deadlock may happen and when to interfere in a coordinated manner for a proper trade-off between deadlock avoidance and original task execution.

In this paper, we present a novel decentralized framework that ensures both efficient task execution and deadlock resolution for multi-agent systems. In particular, we propose a deadlock indicator function for individual agent that utilizes a
CBF-inspired risk measurement to determine when deadlock resolution is needed. Then by incorporating such indicator function design into a unified decentralized control framework, it allows the robots to (i) operate using their original task-related controllers when they are far away from each other or no deadlock will happen even if they are closer, and (ii) seamlessly switch to the deadlock resolution controller only when necessary.
Instead of functioning as a separate design, the proposed deadlock resolution controller is synthesized into the unified decentralized control framework, whose activation depends on the deadlock condition quantified by our deadlock indicator function. This enables co-optimization between the task-related controller and deadlock resolution controller, yielding smoother robots' motion with improved overall task execution efficiency.
The main contributions are summarized below:
\begin{itemize}
    \item A deadlock indicator function with CBF-inspired risk measurement is proposed to adaptively determine when the deadlock resolution needs to be activated/deactivated;
    \item A novel generalized decentralized multi-agent control framework is proposed to
    seamlessly balance deadlock avoidance and task execution;
    \item The detailed theoretical analysis and proof are provided to justify the properties of the proposed approach.
    Simulation and real-world experiment results are also given to demonstrate the performance of the algorithms.
\end{itemize}

\section{Backgrounds}
\subsection{Preliminaries}

Consider a decentralized multi-agent system, where $N$ agents working in a shared 
$k$-dimension
workspace and the dynamics of each agent $i \in \mathcal{I} = \{1, \cdots, N\}$ in control affine form can be described as follows.
\begin{equation} \label{eq:ctlAffine}
    \dot{\mathbf{x}}_i = f(\mathbf{x}_i) + g(\mathbf{x}_i) \mathbf{u}_i
\end{equation}
where $\mathbf{x}_i \in \mathcal{X}_i \subset \mathbb{R}^d$ and $\mathbf{u}_i \in \mathcal{U}_i \subset \mathbb{R}^m$ are the agent's system state and control input, respectively. 
$f: \mathbb{R}^d \mapsto \mathbb{R}^d$ and 
$g: \mathbb{R}^d \mapsto \mathbb{R}^{d \times m}$ are assumed to be locally Lipschitz continuous and full rank \cite{reis2020control}.
Given the control affine system in Eq.~\eqref{eq:ctlAffine} for every agent, we can define a set of stabilizing controls and safety states using the following lemmas.

\begin{lemma}\label{lem:clf}
\label{eq:clf}[Summarized from \cite{ames2019control}] A positive definite function $V(\mathbf{x}_i): \mathbb{R}^d \mapsto \mathbb{R}_{\geq 0}$ is a control Lyapunov function (CLF),
if there exists an extended class-$\mathcal{K}$ function $\gamma(\cdot)$ such that  $\inf_{\mathbf{u}_i \in \mathbb{R}^m}\{\dot{V}(\mathbf{x}_i, \mathbf{u}_i)\}\leq -\gamma(V(\mathbf{x}_i))$. The admissible control space for any Lipschitz continuous controller $\mathbf{u}_i \in \mathbb{R}^m$ enforcing system's stability to $V(\cdot)=0$ thus becomes:
{\small \begin{align}\label{eq:clf_lemma}
    K_\mathrm{clf}(\mathbf{x}_i) = \{ \mathbf{u}_i \in \mathbb{R}^m | L_f V(\mathbf{x}_i) + L_gV(\mathbf{x}_i)\mathbf{u}_i \leq -\gamma(V(\mathbf{x}_i))\}
\end{align}}
where $L_f V(\mathbf{x}_i) = \nabla V^{T}(\mathbf{x}_i) f(\mathbf{x}_i)$ and $L_g V(\mathbf{x}_i) = \nabla V ^{T}(\mathbf{x}_i) g(\mathbf{x}_i)$, respectively.
\end{lemma}
This definition allows us to devise a set of controls to ensure the closed-loop control system's stability prescribed by $V(\cdot)$ for each agent in the multi-agent system.

Denote $\mathbf{x} = [\mathbf{x}_1, \cdots, \mathbf{x}_N]^{T} \in \mathcal{X} \subset \mathbb{R}^{dN}$ as the joint multi-agent states. For any pairwise inter-agent collision avoidance between agent
$i,j \in \mathcal{I}$, the safety condition of $\mathbf{x}$ can be defined as follows.
\begin{align}
h_{ij}(\mathbf{x}) &= \| \pi(\mathbf{x}_i) - \pi(\mathbf{x}_j) \|^{2} - (r_i + r_j) ^2, \forall j \neq i \label{eq:hij}\\
\mathcal{H}_{ij} &= \{\mathbf{x} \in \mathbb{R}^{dN} | h_{ij}(\mathbf{x}) \geq 0\}, \forall j \neq i
\end{align}
where $r_i, r_j \in \mathbb{R}$ are the safety radius of the agent $i$ and $j$, and
$\pi:\mathbb{R}^d \mapsto \mathbb{R}^{s}$
is the mapping from the agent's state to its position in the Cartesian coordinate system.
To guarantee the safety of the decentralized multi-agent system, Control Barrier Functions (CBFs) \cite{ames2014control, wang2017safety, luo2020multi} can be used to derive the set of controls that render the forward invariance of the safe set
for pair-wise agents, as discussed below.

\begin{lemma}\label{lem:de_cbf}[Summarized from \cite{ames2014control, lyu2023risk}]
Given a decentralized multi-agent system with each agent's dynamics defined by Eq.~\eqref{eq:ctlAffine} and a safe set $\mathcal{H}_{ij}$ as the 0-super level set of a continuously differentiable function $h_{ij}(\mathbf{x}): \mathcal{X}\mapsto \mathbb{R}$ for pair-wise agents $i$ and $j$, 
the function $h_{ij}(\mathbf{x})$ is called a control barrier function if there exists an extended class-$\mathcal{K}$ function $\alpha(\cdot)$ such that $\sup_{\mathbf{u}_i \in \mathbb{R}^m}\{\dot{h}_{ij}(\mathbf{x}, \mathbf{u})\}\geq -\alpha(h_{ij}(\mathbf{x}))$.
The admissible control space for any Lipschitz continuous
controller $\mathbf{u}_i \in \mathcal{U}_i$ enforcing the safe set forward invariant thus becomes:
{\small
\begin{align}\label{eq:de_cbc_lemma}
    K_\mathrm{cbf}(\mathbf{x}_i) = \{ \mathbf{u}_i\in \mathcal{U}_i | &L_f h_{ij}(\mathbf{x}) + L_gh_{ij}(\mathbf{x})\mathbf{u}_i  \geq \notag \\
&\qquad \qquad  - W_i \alpha(h_{ij}(\mathbf{x})), \forall j \neq i \}
\end{align}}
where $L_f h_{ij}(\mathbf{x}) = 
\nabla_{\mathbf{x}_i} h_{ij}^{T}(\mathbf{x}) 
f(\mathbf{x}_i)$ and
$L_g h_{ij}(\mathbf{x}) =  \nabla_{\mathbf{x}_i} h_{ij}^{T}(\mathbf{x})g(\mathbf{x}_i)$, respectively.
$W_i \in (0,1)$ is Pairwise Responsibility Weight with $W_i + W_j = 1$.
\end{lemma}

Based on the simplest form of the constraint-driven control \cite{ames2014control, notomista2022resilience} for a decentralized multi-agent system in Eq.~\eqref{eq:ctlAffine} as well as the given CLF and CBF in Lemma~\ref{lem:clf} and Lemma~\ref{lem:de_cbf}, we can define the control objective of stabilizing the system and avoiding inter-robot collisions for agent $i$ as the following Quadratic Programming (QP) problem:
\begin{align} \label{eq:de_frame}
& \qquad \qquad \min _{(\mathbf{u}_i,\delta_i)\in \mathbb {R}^{m+1}} \| {\mathbf{u}_i} \|^{2} + p \delta_i ^{2}\\
&\mathrm{s}.\mathrm{t}.\quad  L_{f} V(\mathbf{x}_i) + L_{g} V(\mathbf{x}_i) \mathbf{u}_i \le -\gamma (V(\mathbf{x}_i)) +  \delta_i \tag{CLF} \\ 
& L_{f} h_{ij}(\mathbf{x}) + L_{g} h_{ij}(\mathbf{x}) \mathbf{u}_i \ge - W_{i} \alpha (h_{ij}(\mathbf{x})), \forall j \neq i \tag{CBF}
\end{align}
where $p > 0$ and $\delta_i \in \mathbb{R}$ is the relaxation slack variable which is used to soften the CLF constraints and make the QP feasible \cite{reis2020control}. The cost function aims to minimize the agent's control effort expended by the system. The CLF constraint is designed to move the agent toward its task-related goal, while the CBF constraint is for regulating
the
controller to keep the system safe, e.g., free from collisions.

However, following a similar analysis as \cite{reis2020control} for the closed-loop control framework in Eq.~\eqref{eq:de_frame}, there exist closed-loop equilibria for the system, including the undesirable ones leading to deadlocks, e.g., agents may be asymptotically stable at states that prevent them from converging to the task-prescribed goals. We summarize the equilibria and the conditions in the following lemma\footnote{For ease of notation, the dependency of $f(\mathbf{x}_i), V(\mathbf{x}_i), L_fV(\mathbf{x}_i), L_gV(\mathbf{x}_i)$, and $g(\mathbf{x}_i)$ on $\mathbf{x}_i$ is omitted, and the dependency of
$h_{ij}(\mathbf{x}), L_fh_{ij}(\mathbf{x})$, and
$L_gh_{ij}(\mathbf{x})$ on $\mathbf{x}$ is omitted.}.

\begin{lemma} \label{lem:deadlock}
[Summarized from \cite{reis2020control}] 
For the control framework in Eq.~\eqref{eq:de_frame} for the decentralized multi-agent system, the closed-loop equilibrium points set is given by $
\mathcal{S} = \mathcal{S}_{\mathrm{int}} \bigcup \mathcal{S}_{\partial \mathcal{H}_{ij}}$, where
{
\footnotesize
\begin{align} 
\mathcal {S}_{\mathrm {int}}=&\{ \mathbf{x}_i \in {\boldsymbol{\Omega }_{\overline {\mathbf {cbf}}}^{\mathbf {clf}}} \cap \mathrm {int}(\mathcal {H}_{ij}) \,|\;\, 
f
= p \gamma (V) G \nabla V \} \\ 
\mathcal {S}_{\partial \mathcal {H}_{ij}}=&\{ \mathbf{x}_i \in {\mathbf {\Omega _{cbf}^{clf}}} \cap \partial \mathcal {H}_{ij} \,|\;\, 
f
= \lambda _{1} G \nabla V - \lambda _{2} G \nabla_{\mathbf{x}_i} h_{ij} \} \label{eq:boundary}
\end{align}}
\end{lemma}
$G=gg^{T}\in\mathbb{R}^{d\times d}$
and 
if $\nabla_{\mathbf{x}_i} h_{ij} \neq 0$, $\lambda _1,\lambda _2$ are given by: 
{\footnotesize
\begin{align*} \lambda _{1}=&\Delta ^{-1} \left ({F_{h_{ij}} \nabla V^{{T}} G \nabla_{\mathbf{x}_i} h_{ij} - F_{V} \| {L_{g} h_{ij}} \|^{2} }\right) \\ \lambda _{2}=&\Delta ^{-1} \left ({F_{h_{ij}} \left ({\| {L_{g} V} \|^{2} + p^{-1} }\right) - F_{V} \nabla V^{{T}} G \nabla_{\mathbf{x}_i} h_{ij} }\right) \end{align*}}
with $F_V = L_fV + \gamma(V)$, $F_{h_{ij}} = L_fh_{ij}+\alpha(h_{ij})$ and
{\footnotesize
\begin{equation*} \Delta = (\nabla V^{{T}} G \nabla h_{ij})^{2} - (p^{-1} + \| {L_{g} V} \|^{2}) \| {L_{g} h_{ij}} \|^{2}
\end{equation*}}
or if $\nabla_{\mathbf{x}_i} h_{ij} = 0$, then $\lambda_1 = F_V/(1/p+ \|LgV\|^{2})$, 
and Eq.~\eqref{eq:boundary} holds for any non-negative $\lambda_2$.

$\mathcal{S}_{\mathrm{int}}$ and $\mathcal{S}_{\partial \mathcal{H}_{ij}}$ are the sets of interior equilibria and boundary equilibria, respectively, with
{ \footnotesize
\begin{align*} {\boldsymbol{\Omega }_{\overline {\mathbf {cbf}}}^{\mathbf {clf}}}=&\left\{{ x \in \mathbb {R}^{d}~:~\frac {\nabla V^{{T}} G \nabla_{\mathbf{x}_i} h_{ij}}{p^{-1} + \| {L_{g} V} \|^{2}} F_{V} < F_{h_{ij}} , F_{V} \ge 0 }\right\} \quad \\ {\mathbf {\Omega _{cbf}^{clf}}}=&\{ x \in \mathbb {R}^{d}~:~\lambda _{1}, \lambda _{2} \ge 0 \} \end{align*}}

Specifically, the set ${\boldsymbol{\Omega }_{\overline {\mathbf {cbf}}}^{\mathbf {clf}}}$ denotes the states where the CLF constraint is active and the CBF constraint is inactive in Eq.~\eqref{eq:de_frame}, while ${\mathbf {\Omega _{cbf}^{clf}}}$ denotes the states where both CLF and CBF constraints are active.

While both interior and boundary equilibria exist, only the interior equilibria in $\mathcal{S}_{\mathrm{int}}$ align with the objective of the CLF as defined in~Eq.~\eqref{eq:de_frame}. In contrast, the boundary equilibria in $\mathcal{S}_{\partial \mathcal{H}_{ij}}$, which is located at the boundary of the safe set as defined in~Eq.~\eqref{eq:boundary},
can prevent the system from converging to the desired task-related equilibrium where $V(\cdot)=0$. This can cause agents to stall, leading to deadlock \cite{grover2023before}.

\subsection{Problem Statement}

In this paper, our goal is to design a deadlock-free controller for decentralized multi-agent systems that ensures asymptotic stability while guaranteeing collision avoidance. 
Although \cite{reis2020control} proposed one solution by introducing an auxiliary CBF to prevent the system from converging to the boundary equilibria
regardless of when that may occur, it could
dominate the agents’ behaviors from the starting stage and thus unnecessarily reduce the task execution efficiency.
To address deadlock avoidance while accounting for original task efficiency,
the problem becomes how to develop a generalized decentralized control framework that co-optimizes the control objectives such as deadlock resolution, collision avoidance, and task-related stabilization, so that the decentralized multi-agent system can navigate towards the task-related goals without collisions while adaptively avoiding deadlock conditions in Eq.~\eqref{eq:de_frame} when necessary.

\section{Method}
In this section, we present a novel control framework subject to the deadlock problem inherent in the closed-loop control framework in Eq.~\eqref{eq:de_frame} for a decentralized multi-agent system. Importantly, in an effort to improve task efficiency, we will explore a generalized framework that will adaptively decide when to unlock/deactivate the deadlock resolution according to the observed neighboring agents' information.

\subsection{Deadlock Indicator Function}
The CBF-inspired risk metric proposed in \cite{lyu2023risk} has been demonstrated to effectively reflect the dynamic proximity of agents in decentralized multi-agent systems,
which is defined as follows for each agent $i$.
\begin{equation} \label{eq:risk}
R_i = \frac{1}{N-1}\sum_{j \neq i} (-\dot{h}_{ij}(\mathbf{x}, \mathbf{u}_i)-\alpha (h_{ij}(\mathbf{x}))) + \phi
\end{equation}
where $\phi$ is a positive constant to ensure that $R_i$ is always positive. 
A higher value of the CBF-inspired risk measurement $R_i$ indicates that agent $i$ is operating in a denser environment and thus more likely to collide with other agents 
accounting for their proximity and executed control from the last time step. Note that in this paper, we assume that all the other agents' states and velocities are considered for risk evaluation.

Motivated by the CBF-inspired risk measurement in Eq.~\eqref{eq:risk}, we further propose a deadlock indicator function $\zeta_i(R_i)$ for agent $i$ as:
\begin{equation} \label{eq:sigmoid}
    \zeta_i(R_i) =  \frac{1}{1+e^{-t(R_i-c)}}
\end{equation}
where $c$ and $t$ are user-specified constant. When agent $i$ is far from other agents or obstacles, the CBF-inspired risk measurement $R_i$ is small, rendering $\zeta_i(R_i) \rightarrow 0$. In this case, the agent perceiving a low-risk environment is able to proceed toward its goal without significant concern for a deadlock, as it remains far from boundary equilibria where deadlock situations are less likely to occur. 
On the other hand, when the agent $i$ is close to other agents or obstacles, the value of CBF-inspired risk measurement $R_i$ becomes larger, leading to $\zeta_i(R_i) \rightarrow 1$, and the framework must start to account for potential deadlock situations by initiating deadlock resolution to ensure safe and efficient navigation.

\subsection{A Generalized Deadlock Resolution Framework}

Based on the discussion above, the deadlock indicator function can thus serve as a switch to adaptively determine when the deadlock resolution should be considered. In this section, we present how to integrate the deadlock indicator function into the deadlock resolution and formulate a generalized decentralized multi-agent control framework which seamlessly balances deadlock avoidance with task execution.

Similar to \cite{reis2020control}, we introduce an additional rotation
for the state $\mathbf{x}_i$ of agent $i$ and an auxiliary CBF to avoid deadlock by temporarily changing the task-related CLF.

Firstly, an extended CLF
$V_q: \mathbb{R}^d \times \mathcal{SO}(d) \mapsto \mathbb{R}$
for the task-related CLF\footnote{We assume that $V(\cdot)$ is a nonradial function so that the rotation matrix will introduce a difference to the agent's state, because if $V$ is a radial function, $V(\mathbf{Q}_i\mathbf{x}_i) = V(\mathbf{x}_i)$} $V(\cdot)$ is introduced as follows.
{\small\begin{equation} \label{eq:Vq}
    V_q(\mathbf{x}_i, \mathbf{Q}_i) = V(\mathbf{Q}_i\mathbf{x}_i)
\end{equation}}
where $\mathbf{Q}_i \in \mathcal{SO}(d)$ represents a rotation matrix from the special orthogonal group that rotates the agent’s state $\mathbf{x}_i$. The derivative of $V_q$ is given by:
{\small
\begin{equation}
    \dot{V}_q(\mathbf{x}_i, \mathbf{Q}_i) = (\mathbf{Q}_i^{T}\nabla V(\mathbf{Q}_i\mathbf{x}_i))^{T}\dot{\mathbf{x}}_i + (\mathcal{O}_d(\mathbf{x}_i)^{T}\mathbf{Q}_i^{T}\nabla V(\mathbf{Q}_i\mathbf{x}_i))^{T}\omega_i \notag
\end{equation}}
where $\omega_i \in \mathbb{R}^{\frac{1}{2} d(d-1)}$ is the control signal that governs the agent’s state rotation. Specifically, the rotational dynamics of the agent’s state are determined by the evolution of $\mathbf{Q}_i$, which satisfies the differential equation
$\dot{\mathbf{Q}}_i = \mathbf{Q}_i \hat{\omega}_i$,
where $\hat{\omega}_i \in \mathfrak{so}(d)$ is a skew-symmetric matrix from the Lie algebra of the rotation group, constructed from $\omega_i$. 
The operator $\mathcal{O}_d : \mathbb{R}^d \rightarrow \mathbb{R}^{d \times \frac{d}{2} (d-1)}$ is defined by the
relation $\hat{\omega}_i \mathbf{x}_i = \mathcal{O}_d(\mathbf{x}_i)\omega_i$, capturing the interaction between the control signal and the agent's state.

In this way, $V_q$ introduces a rotation to the agent state so that the original task goal is virtually deviated and thus helps to solve the deadlock problem.
With this, we further integrate our deadlock indicator function with the $V$ and $V_q$ to construct our CLF constraint that drives the agent toward its goal with adaptive deadlock resolution as follows.
{\small\begin{align}
& (1 - \zeta_i(R_i))(\dot{V}(\mathbf{x}_i, \mathbf{u}_i) + \gamma (V(\mathbf{x}_i))) + \notag \\
& \zeta_i(R_i) (\dot{V}_q(\mathbf{x}_i, \omega_i, \mathbf{u}_i) + \gamma (V_q(\mathbf{x}_i, \mathbf{Q}_i))) \le \delta_i \label{eq:CLF}
\end{align}}
\begin{remark}\label{remark:CLF}
By integrating our deadlock indicator function into the constraint to formulate a CLF constraint, we dynamically adjust the task function as needed. 
Specifically, when $\zeta_i(R_i) \rightarrow 0$), Eq.~\eqref{eq:CLF}
simplifies to
Eq.~\eqref{eq:de_frame}, 
which prioritizes task execution as the agent is confident that it is not at risk of encountering deadlocks.
When $\zeta_i(R_i) \in (0,1]$, Eq.~\eqref{eq:CLF} enforces the agent to deviate from its original task goal.
\end{remark}

Secondly, one necessary condition for the boundary equilibria condition in Eq.~\eqref{eq:boundary} is that the vectors $f, G \nabla V_q$ and $G\nabla_{\mathbf{x}_i} h_{ij}$ are simultaneously collinear, i.e., these vectors are parallel.
We follow \cite{reis2020control} to use the positive semi-definite function $\mathcal {D}_{ij}$ to evaluate the proximity of the system state from the collinearity condition.
{\footnotesize
\begin{equation} \label{eq:colinear}
\mathcal {D}_{ij}(\mathbf{x},\mathbf{Q}_i) = \frac {1}{2} \nabla V_q(\mathbf{x}_i,\mathbf{Q}_i)^{T} G \left ({\mathcal {P}_{f} + \mathcal {P}_{G \nabla_{\mathbf{x}_i} h_{ij}} } \right) G \nabla V_q(\mathbf{x}_i,\mathbf{Q}_i) \notag
\end{equation}}
where $\mathcal{P}$ is a scaled orthogonal projection operator. For a vector $v \in \mathbb{R}^n$, $\mathcal{P}_v = ||v||^2 \mathbf{I} - vv^{T} \in \mathbb{R}^{n \times n}$. In this way, the necessary condition for the boundary equilibria condition, i.e., the vectors $f, G \nabla V_q$ and $G\nabla_{\mathbf{x}_i} h_{ij}$ are simultaneously collinear and thus become $\mathcal{D}_{ij}=0$. To avoid such condition, we use the auxiliary CBF $h_{\mathcal{D}_{ij}}$ from \cite{reis2020control} to enforce $\mathcal{D}_{ij}>0$.
\begin{equation}\label{eq:h_dij}
    h_{\mathcal{D}_{ij}}(\mathbf{x}, \mathbf{Q}_i) = \psi(h_{ij}(\mathbf{x}))(\mathcal{D}_{ij}(\mathbf{x}, \mathbf{Q}_i) - \epsilon)
\end{equation}
where $\epsilon$ is a small positive constant and $\psi(\cdot)$ is a smooth and decreasing function with (1) $\psi(\cdot) > 0$, (2) $\psi^{'}(0) = 0 $ and (3) $ \lim_{h \rightarrow \infty}\psi(h) = 0$.
The derivative of $h_{\mathcal{D}_{ij}}$ is given by:
$\dot{h}_{\mathcal{D}_{ij}} = L_f h_{\mathcal{D}_{ij}} + L_g h_{\mathcal{D}_{ij}} \mathbf{u}_i + \nabla_{Qi} h_{\mathcal{D}_{ij}}^{T} \omega_i$ with $L_f h_{\mathcal{D}_{ij}} = \nabla h_{\mathcal{D}_{ij}}^{T} f(\mathbf{x}_i)$ and $L_g h_{\mathcal{D}_{ij}} = \nabla h_{\mathcal{D}_{ij}}^{T}g(\mathbf{x}_i)$ where 
{\footnotesize\begin{align}
&\nabla h_{\mathcal {D}_{ij}}(\mathbf{x},\mathbf{Q}_i)=\psi (h_{ij}) \nabla \mathcal {D}_{ij} + \sigma '(h_{ij}) (\mathcal {D}_{ij} - \epsilon) \nabla_{\mathbf{x}_i} h_{ij} \\
&\nabla _{Q} h_{\mathcal {D}_{ij}}(\mathbf{x},\mathbf{Q}_i)=\psi (h_{ij}) \nabla _{Qi} \mathcal {D}_{ij} 
\end{align}}
where $\nabla \mathcal {D}_{ij}$ and $\nabla _{Qi} \mathcal {D}_{ij}$ are given as follow.
{\footnotesize\begin{align} \label{eq:D}
\nabla \mathcal {D}_{ij}= &(H_{V} G + \Gamma ^{T}_{g,\nabla V}) (\mathcal {P}_{f} + \mathcal {P}_{G \nabla_{\mathbf{x}_i} h_{ij}}) G \nabla V +(H_{h_{ij}} G \notag \\
&+ \Gamma ^{{T}}_{g,\nabla h_{ij}}) \mathcal {P}_{G \nabla V} G \nabla_{\mathbf{x}_i} h_{ij} + \nabla f^{{T}} \mathcal {P}_{G \nabla V} f 
\end{align}}
{\footnotesize
\begin{equation}
\nabla _{Qi} \mathcal {D}_{ij} = (H_{V} \mathcal {O}_{d}(\mathbf{x}_i) - \mathcal {O}_{d}(\nabla V))^{{T}} G (\mathcal {P}_{f} + \mathcal {P}_{G \nabla_{\mathbf{x}_i} h_{ij}}) G \nabla V 
\end{equation}}
$H_{V}$ and $H_{h_{ij}}$ represent the Hessian matrices of the CLF and the CBF.
$\Gamma_{a, b}(x) \in \mathbb{R}^{d \times d}$ is defined as $\Gamma _{a, b}(x) = \Sigma_{i=1}^{m}(a_{i}(x)^{T}b\mathbf{I} + a_{i}(x)b^{T})\nabla a_{i}$ for the matrix function $a(x) = [a_{1}(x), \cdots, a_{m}(x)] \in \mathbb{R}^{d \times m}$ with a vector $b \in \mathbb{R}^d$.

We further integrate our deadlock indicator function into CBF in Eq.~\eqref{eq:h_dij} to construct the CBF constraint as:
\begin{equation}\label{eq:CBF2}
\zeta_i(R_i)(\dot{h}_{\mathcal{D}_{ij}}(\mathbf{x}, \omega_i, \mathbf{u}_i) + \beta (h_{\mathcal {D}_{ij}}(\mathbf{x}, \mathbf{Q}_i))) \ge 0
\end{equation}
where $\beta(\cdot)$ is a class-$\mathcal{K}$ function. By integrating the deadlock indicator function into the CBF constraint in Eq.~\eqref{eq:CBF2}, it is straightforward that when the agents are far away from each other (i.e., $\zeta_i(R_i) \rightarrow 0$), the constraint would directly hold. In this case, compared to the constraints in \cite{reis2020control} (i.e., $\dot{h}_{\mathcal{D}_{ij}}(\mathbf{x},\omega_i, \mathbf{u}_i) + \beta (h_{\mathcal{D}_{ij}}(\mathbf{x}), \mathbf{Q}_i) \ge 0$), the proposed constraints in Eq.~\eqref{eq:CBF2} remain inactive in a dynamic environment unless the agents are in close proximity. This allows agents to autonomously detect deadlock and respond accordingly.

\begin{theorem}
Consider the multi-agent system with the agent's dynamics given by Eq.~\eqref{eq:ctlAffine}. Assume we have the task related CLF $V(\mathbf{x}_i)$, the extended CLF $V_q(\mathbf{x}_i, \mathbf{Q}_i)$ given by Eq.~\eqref{eq:Vq}, the CBF $h_{ij}(\mathbf{x})$ given by Eq.~\eqref{eq:hij}, and CBF $h_{\mathcal{D}_{ij}}(\mathbf{x}, \mathbf{Q}_i)$ given by Eq.~\eqref{eq:CBF2}. Then for any initial condition that $\mathcal{D}_{ij} \neq 0$,
the solution of the following QP renders deadlock-free trajectories.
\begin{subequations}\begin{align} \label{eq:prop}
&\min _{ \!\!\substack {(\mathbf{u}_i,\omega_i,\delta_i) \in \mathbb {R}^{m} \times \mathbb {R}^{\frac {d}{2}(d-1)} \times \mathbb {R} }\!\!}~ \| {\mathbf{u}_i} \|^{2} + q \| {\omega_i } \|^{2} + p \delta_i ^{2} \tag{20} \\
& \mathrm{s}.\mathrm{t}.~~ (1 - \zeta_i(R_i))(\dot{V}(\mathbf{x}_i, \mathbf{u}_i) + \gamma (V(\mathbf{x}_i))) \notag \\
& \qquad + \zeta_i(R_i) (\dot{V}_q(\mathbf{x}_i, \omega_i, \mathbf{u}_i) + \gamma (V_q(\mathbf{x}_i, \mathbf{Q}_i))) \le \delta_i \label{CLF}
\\
&\qquad \dot{h}_{ij}(\mathbf{x}, \mathbf{u}_i) + W_i \alpha (h_{ij}(\mathbf{x})) \ge 0, \forall j \neq i \label{CBF1}
\\
&\qquad \zeta_i(R_i)(\dot{h}_{\mathcal{D}_{ij}}(\mathbf{x}, \omega_i, \mathbf{u}_i) + \beta (h_{\mathcal {D}_{ij}}(\mathbf{x}))) \ge 0, \forall j \neq i  \label{CBF2}
\end{align} \end{subequations}
where $p, q > 0 \in \mathbb{R}$ are user-defined constants. 

\end{theorem}

\begin{proof}
\textbf{Case 1}: when $\zeta_i(R_i) \to 0$, the framework will be a decentralized CLF-CBF in Eq.~\eqref{eq:de_frame}. Because the agents are far away from each other, the CBF constraint in Eq.~\eqref{eq:de_frame} is inactive. 
Consequently, no boundary equilibria exist, and therefore, deadlock is avoided.

\textbf{Case 2}: when $\zeta_i(R_i) \to 1$, the CLF constraint in Eq.~\eqref{CLF} would boil down to:
$\dot{V}_q(\mathbf{x}_i, \omega_i, \mathbf{u}_i) + \gamma (V_q(\mathbf{x}_i, \mathbf{Q}_i)) \le \delta_i$.
And the CBF constraint in Eq.~\eqref{CBF2} would boil down to:
$\dot{h}_{\mathcal{D}_{ij}}(\mathbf{x}, \omega_i, \mathbf{u}_i) + \beta (h_{\mathcal {D}_{ij}}(\mathbf{x}, \mathbf{Q}_i)) \ge 0$.
Thus the framework would boil down to the CLF-CBF with deadlock resolution method Theorem. 3 in \cite{reis2020control}. One can follow \cite{reis2020control} for detailed proof of the deadlock-free guarantee.

\textbf{Case 3}: when $\zeta_i(R_i) \in (0,1)$, the Lagrangian associated with the QP in Eq.~\eqref{eq:prop} for agent $i$ is,
{\footnotesize\begin{align*} 
\mathcal {L} &=  \| {\mathbf{u}_i} \|^{2} +  p \delta_i ^{2} + q  \|\omega_i  \|^ 2 - \lambda _{2} (F_{h_{ij}} + L_{g} h_{ij} \, \mathbf{u}_i)\\
&+ \lambda _{1}((1- \zeta_i) (F_{V} + L_{g} V \mathbf{u}_i) + \zeta_i (F_{Vq} + L_{g} V_q \mathbf{u}_i)- \delta_i) \\
& - \lambda_3 \zeta_i (L_f h_{\mathcal{D}_{ij}} + L_g h_{\mathcal{D}_{ij}} \mathbf{u}_i + \nabla_{Qi} h_{\mathcal{D}_{ij}}^{T} \omega_i + \beta{h_{\mathcal{D}_{ij}}})
\end{align*}}
where $\lambda_1, \lambda_2, \lambda_3$ are the Karush-Kuhn-Tucker (KKT) multipliers, and $\zeta_i = \zeta_i(R_i)$. 
The KKT conditions are:
{\footnotesize\begin{align}
&\frac {\partial \mathcal {L}}{\partial \mathbf{u}_i} = 2\mathbf{u}_i + \lambda _{1} (1-\zeta_i) L_{g} V^{{T}} + \lambda_1 \zeta_i L_g V_g^{T} \notag \\
&\quad \quad   - \lambda _{2} L_{g} h_{ij}^{{T}} - \lambda_3 \zeta_i L_gh_{\mathcal{D}_{ij}}= 0 \label{eq:kkt1} \\
&\frac {\partial \mathcal {L}}{\partial \delta_i } = 2 p \delta_i - \lambda _{1} = 0 \label{eq:delta_i} \\
&\frac {\partial \mathcal {L}}{\partial \omega_i } = 2 q \omega_i^{T} + \lambda _{1}\zeta_i(\mathcal{O}_d(\mathbf{x}_i)^{T}\mathbf{Q}_i^{T}\nabla V(\mathbf{Q}_i\mathbf{x}_i))^{T} = 0 \notag \\
& \lambda _{1}((1- \zeta_i) (F_{V} + L_{g} V \mathbf{u}_i) + \zeta_i (F_{Vq} + L_{g} V_q \mathbf{u}_i)- \delta_i) = 0 \notag \\
&\lambda _{2} (F_{h_{ij}} + L_{g} h_{ij} \,\mathbf{u}_i) = 0 \notag  \\
&\lambda_3 \zeta_i (L_f h_{\mathcal{D}_{ij}} + L_g h_{\mathcal{D}_{ij}} \mathbf{u}_i + \nabla_Q h_{\mathcal{D}_{ij}}^{T} \omega_i + \beta({h_{\mathcal{D}_{ij}}})=0 \notag
\end{align}}

Because the CBF constraint in Eq.~\eqref{CBF2} is used to avoid the undesired boundary equilibria conditions, we need to consider the possible solutions of the QP depending on the activation of the constraint Eq.~\eqref{CBF2}.

\textit{Subcase 1}
(if the possible solution would make the system get into a deadlock when the constraint Eq.~\eqref{CBF2} is inactive): Assume by contradiction that there exists a boundary equilibrium point $(\mathbf{x}_i^{\star}, \mathbf{Q}_i^{\star})$ for the agent $i$. Because Eq.~\eqref{CBF2} is inactive,
the boundary equilibrium condition in 
Theorem~\ref{lem:deadlock} \cite{reis2020control} holds, i.e., 
$f$, $G\nabla V$, and $G \nabla_{\mathbf{x}_i} h_{ij}$ are collinear at the boundary equilibrium point ($\mathcal{D}(\mathbf{x}_i^{\star}, \mathbf{Q}_i^{\star}) = 0$). It contradicts the setting that $h_{\mathcal{D}_{ij}} > 0$ in Eq.~(\ref{eq:h_dij}) is forward invariant.

\textit{Subcase 2} (if the possible solution would make the system get into a deadlock when the constraint Eq.~\eqref{CBF2} is active, i.e., $\dot{h}_{\mathcal{D}_{ij}} + \beta(h_{\mathcal{D}_{ij}}) = 0$):
Assume by contradiction that there exists a boundary equilibrium point $(\mathbf{x}_i^{\star}, \mathbf{Q}_i^{\star})$ for the agent $i$ which would make the agent fall into the deadlock. Because $(\dot{\mathbf{x}}_i, \dot{\mathbf{Q}}_i) = (0,0)$ at the equilibrium point, we can get that $\dot{V} = 0$, $\dot{V}_q = 0$, $\dot{h}_{ij} = 0$ and $\dot{h}_{\mathcal{D}_{ij}} = 0$. From the KKT condition in Eq.~\eqref{eq:delta_i}, we have $\delta_i =0.5 p^{-1} \lambda_1$. According to the CLF constraint Eq.~\eqref{CLF}, we can get $(1-\zeta_i)(\gamma(V(\mathbf{x}_i^{\star}))) + \zeta_i \gamma(V_q(\mathbf{x}_i^{\star})) \leq \delta_i$. Furthermore, we can get $\lambda_1 \geq p((1-\zeta_i)(\gamma(V(\mathbf{x}_i^{\star}))) + \zeta_i \gamma(V_q(\mathbf{x}_i^{\star})))$.
Because the functions $V$ and $V_q$ are both positive definite, we have $\lambda_1 > 0$, which means that CLF constraint Eq.~\eqref{CLF} is active. Additionally, since $\mathbf{x}_i^{\star} \in \partial \mathcal{H}$, 
$h_{ij} = 0$. Because $\dot{h}_{ij} + h_{ij} = 0$, we can ensure that the constraint Eq.~\eqref{CBF1} is also active. 
Because Eq.~\eqref{CBF2} is active and $\dot{h}_{\mathcal{D}_{ij}} = 0$, we can say that $h_{\mathcal{D}_{ij}} = \psi(0)(\mathcal{D}_{ij} - \epsilon) = 0$. Therefore, $\mathcal{D}_{ij} = \epsilon$, which means that $f(\mathbf{x}_i^{\star})$, $G(\mathbf{x}_i^{\star})\nabla V(\mathbf{x}_i^{\star}, \mathbf{Q}_i^{\star})$ and $G(\mathbf{x}_i^{\star})\nabla_{\mathbf{x}_i} h_{ij}(\mathbf{x}_i^{\star})$ are not simultaneously collinear. Using similar argument \cite{reis2020control} and according to Eq.~\eqref{eq:kkt1}, the solution for $\mathbf{u}_i$ is:

{\footnotesize
\begin{equation} \label{eq:condi}
\mathbf{u}_i^{\star} = \frac{ \lambda _{2} L_{g} h_{ij}^{{T}} + \lambda_3 \zeta_i L_gh_{\mathcal{D}_{ij}}^{T} - \lambda _{1} (1-\zeta_i) L_{g} V^{{T}} - \lambda_1 \zeta_i L_g V_g^{T}}{2} 
\end{equation}}
When all the constraints are active, the KKT conditions yield the new equilibrium condition:
{\footnotesize\begin{align}
&f + g \mathbf{u}_i^{\star} = 0 \notag \\
&f + g \frac{ \lambda _{2} L_{g} h_{ij}^{{T}} + \lambda_3 \zeta_i L_gh_{\mathcal{D}_{ij}}^{T} - \lambda _{1} (1-\zeta_i) L_{g}V^{{T}} - \lambda_1 \zeta_i V_g^{T}}{2} = 0 \notag \\
&f =\frac{\lambda_1 ((1-\zeta_i) G \nabla V + \zeta_i G \nabla V_g) - \lambda_2 G \nabla_{\mathbf{x}_i} h_{ij} - \lambda_3 \zeta_i \psi(0) G \nabla \mathcal{D}_{ij}}{2} \notag
\end{align}}
Because at the equilibrium point $\omega_i = 0$, we have $\nabla V = \nabla V_g$. Hence, the condition in Eq.~\eqref{eq:condi} will be: 
{\footnotesize
\begin{equation} \label{eq:condition}
f -0.5 \lambda_1 G \nabla V + 0.5\lambda_2 G \nabla_{\mathbf{x}_i} h_{ij} + 0.5 \lambda_3 \zeta_i \psi(0) G \nabla \mathcal{D}_{ij} = 0
\end{equation}}
Since $\lambda_1 > 0$ and $g$ is full rank, the second term needs to be canceled out by other terms. Since $f, G \nabla V$ and $G \nabla h_{ij}$ are not simultaneously collinear, we can summarize possible necessary conditions for Eq.~\eqref{eq:condition} holds:(i) $f || G \nabla_{\mathbf{x}_i} h_{ij}$, $\nabla V || \nabla \mathcal{D}_{ij}$,
(ii) $f || G \nabla V$, $\nabla_{\mathbf{x}_i} h_{ij} || \nabla \mathcal{D}_{ij}$, and
(iii) $\nabla_{\mathbf{x}_i} h_{ij} || \nabla V$, $f || G \nabla \mathcal{D}_{ij}$. (The expression $v || w$
denotes that two vectors $v, w \in \mathbb{R}^{n}$ are parallel.) These conditions can be rewritten as: (i) $f = \textit{k}_1 G \nabla_{\mathbf{x}_i} h_{ij}$, $\mathcal{P}_{\nabla V} \nabla \mathcal{D}_{ij} = 0$,
(ii) $f = \textit{k}_2 G \nabla V$, $\mathcal{P}_{\nabla_{\mathbf{x}_i} h_{ij}} \nabla \mathcal{D}_{ij} = 0$, and
(iii) $\nabla_{\mathbf{x}_i} h_{ij}=  \textit{k}_3 \nabla V$, $\mathcal{P}_{f} G \nabla \mathcal{D}_{ij} = 0$, respectively($\textit{k}_1, \textit{k}_2, \textit{k}_3 \in \mathbb{R}$ are constants). If $\mathcal{P}_{(\cdot)}$ is $0$, then $(\cdot)$ will be collinear with the other two vectors, which contradicts with $\mathcal{D}_{ij} = \epsilon$ i.e., $f, G \nabla V$ and $G \nabla_{\mathbf{x}_i} h_{ij}$ should be not simultaneously collinear. 
If $\mathcal{P}_{(\cdot)} $ is not $0$, then for $\mathcal{P}_{\nabla V} \nabla \mathcal{D}_{ij} = 0$, $\mathcal{P}_{\nabla_{\mathbf{x}_i} h_{ij}} \nabla \mathcal{D}_{ij} = 0$ and $\mathcal{P}_{f} G \nabla \mathcal{D}_{ij} = 0$ holding, we should guarantee $\nabla \mathcal{D}_{ij} = 0$. Thus $f, G \nabla V$ and $G \nabla_{\mathbf{x}_i} h_{ij}$ should be collinear according to Eq.~\eqref{eq:D}, which also contradicts with $\mathcal{D}_{ij} = \epsilon$. Therefore we can conclude that there is no deadlock existing for the system. 
\end{proof}

\begin{remark}
Note that when setting $\zeta_i(R_i) = 1$, our method will degrade to the CLF-CBF with deadlock resolution method proposed in \cite{reis2020control}. Therefore, our method is a more generalized framework that not only addresses the potential deadlock issues but also adaptively
activates or deactivates
the deadlock resolution as necessary.
\end{remark}

\begin{remark}
Although theoretically $\zeta_i(R_i)$ is in the range of $(0,1)$, it could be $0$ and $1$ in practice because of
the limitation of numerical precision, which makes our discussions of \textbf{Case 1} and \textbf{Case 2} reasonable and also correspond to the plots in Fig~\ref{fig:zeta}.
\end{remark}

\begin{remark}
When the constraint in Eq.~\eqref{CBF2} is switched on, it is possible that $\mathcal{D}_{ij}(\mathbf{x}, \mathbf{Q}_i) = 0$, which will make the QP Eq.~\eqref{eq:prop} unfeasible. To this end, a practical way is to set the optimal solution as $\mathbf{u}_i^{\star} = 0$ and $\omega_i^{\star} = \omega_c$, where $\omega_c$ is a small constant \cite{reis2020control}.
\end{remark}

\section{Results}
\subsection{Numerical Results with Point Agents}
In this section, we present the numerical simulation examples to verify the performance of our method using four agents with single integrator dynamics $\dot{\mathbf{x}}_i = \mathbf{u}_i$ for agent $i$. 
We take the CLF as $V(\mathbf{x}_i) = || \mathbf{x}_i - \mathbf{x}^\mathrm{goal}_{i}||^2$, which drives the agent towards its goal,
and CBF as $h_{ij}(\mathbf{x}) = ||\mathbf{x}_i - \mathbf{x}_j||^2 - (r_i + r_j)^2$, which defines the pairwise agents' safety. We compare \textbf{our method} with the \textbf{CLF-CBF} method Eq.~\eqref{eq:de_frame} and the \textbf{CLF-CBF with deadlock resolution} contextualized from \cite{reis2020control} with the same parameters setting. The initial positions for the four agents are $(-10,0), (0,10), (10,0), (0, -10)$.

\begin{figure}[!htbp]\label{fig:1}
\centering
\begin{subfigure}{0.46\linewidth}
  \includegraphics[width = \linewidth]{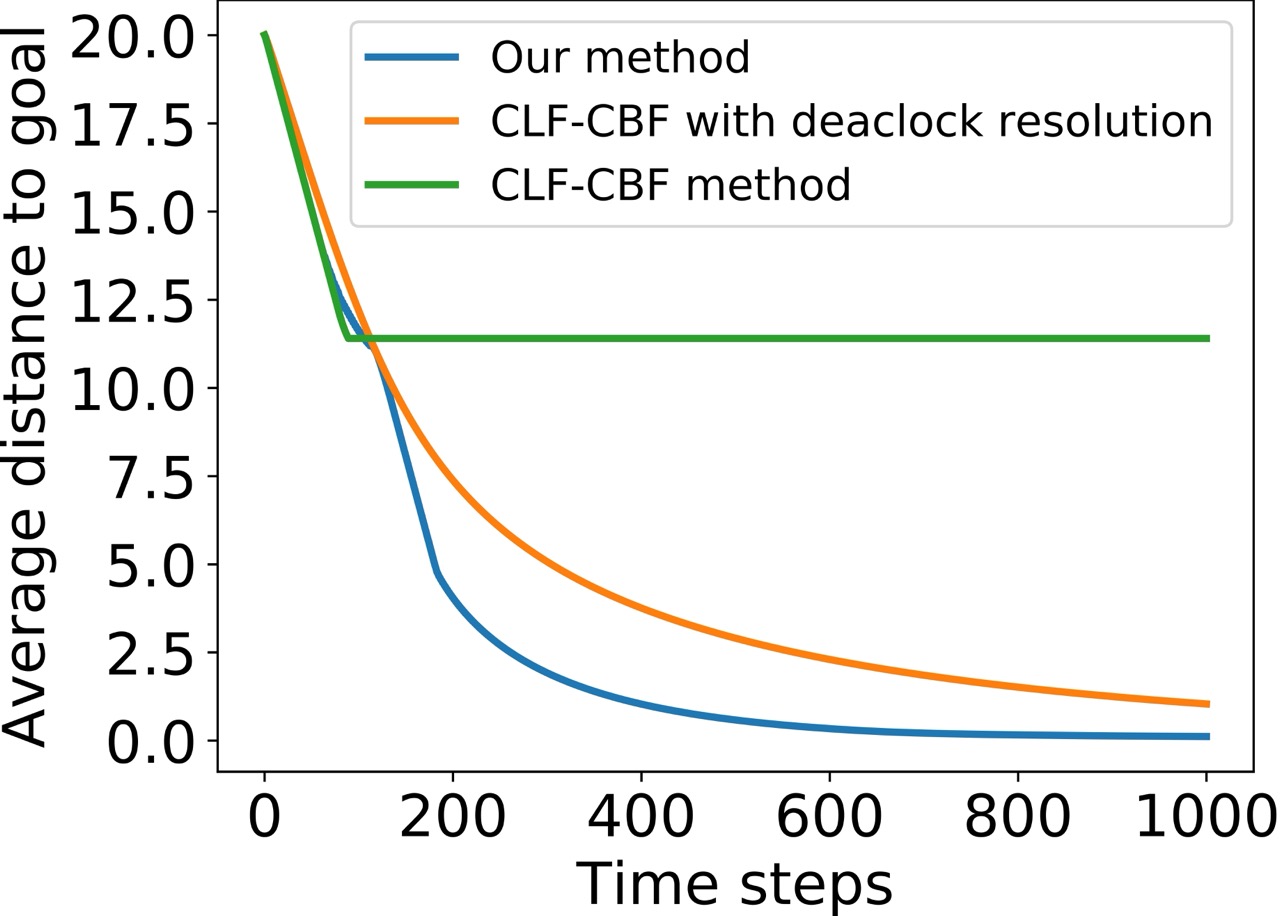}
 \caption{}
  \label{fig:comparison}
   \end{subfigure}
\begin{subfigure}{0.46\linewidth}
  \includegraphics[width = \linewidth]{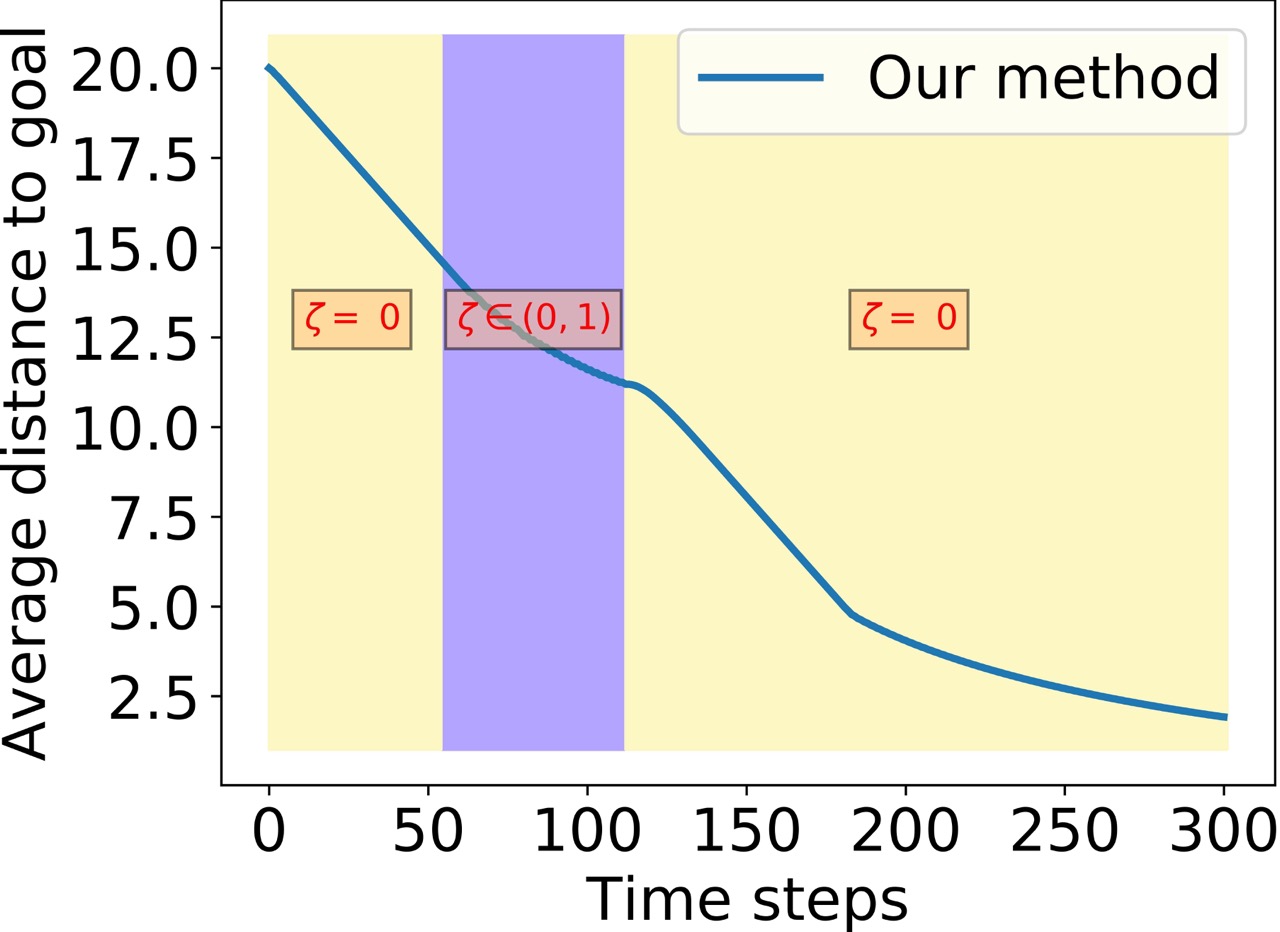}
 \caption{}
  \label{fig:zeta}
\end{subfigure}
\caption{ Average distances from current time step positions to goal positions of four agents. Fig.~\ref{fig:comparison} shows the performance of the CLF-CBF method, CLF-CBF with deadlock resolution method, and our method. Fig.~\ref{fig:zeta} shows the average $\zeta$ value over the initial 300 time steps in Fig~\ref{fig:comparison} with our method, where $\zeta = 0$ means the deadlock resolution is not activated and $\zeta \in (0,1)$ means the deadlock resolution is activated.}
\vspace{-0.3cm}
\end{figure}

Firstly, in Fig. \ref{fig:comparison}, it describes that with the aforementioned three methods, how the average distance from current robot positions to goals changes over time.
The blue line, orange line, and green line represent the average distances using the proposed decentralized control framework, CLF-CBF with deadlock resolution method in~\cite{reis2020control}, and CLF-CBF method, respectively. It is observed that the green line stops dropping around $t=70$, indicating that the agents fall into a deadlock using the CLF-CBF method. Both the CLF-CBF with deadlock resolution method and our method drive the agents toward their goals. However, it is obvious that our method would converge to the goal position earlier, which means that our method is more efficient than the CLF-CBF with deadlock resolution method.
In addition, we plot the average distance from the current robot positions to the goals
with the average $\zeta$ value shown in different periods with our method in Fig.~\ref{fig:zeta}, where $\zeta = \frac{1}{4}\Sigma_{i=1}^{4} \zeta_i$.

Secondly, we compare the agent behavior generated by the CLF-CBF method, CLF-CBF with deadlock resolution method \cite{reis2020control}, and our method in Fig.~\ref{fig:base} and Fig.~\ref{fig:prop}. 
As observed
from the trajectories
in Fig.~\ref{fig:deadlock}, four agents stall at the central parts and cannot head to their goals, which means they are getting into a deadlock. From the  CLF-CBF with deadlock resolution method's trajectory shown in Fig.~\ref{fig:baseline}, we can see the agents start to detour earlier rather than detour when it is necessary. 
With our method in Fig.~\ref{fig:prop}, we can observe that the agents would directly move to the goal when they are safe (far enough from other agents) and start to detour when it is necessary (close to each other). 

\begin{figure}[!htbp]
\centering
\begin{subfigure}{0.45\linewidth}
  \includegraphics[width = \linewidth]{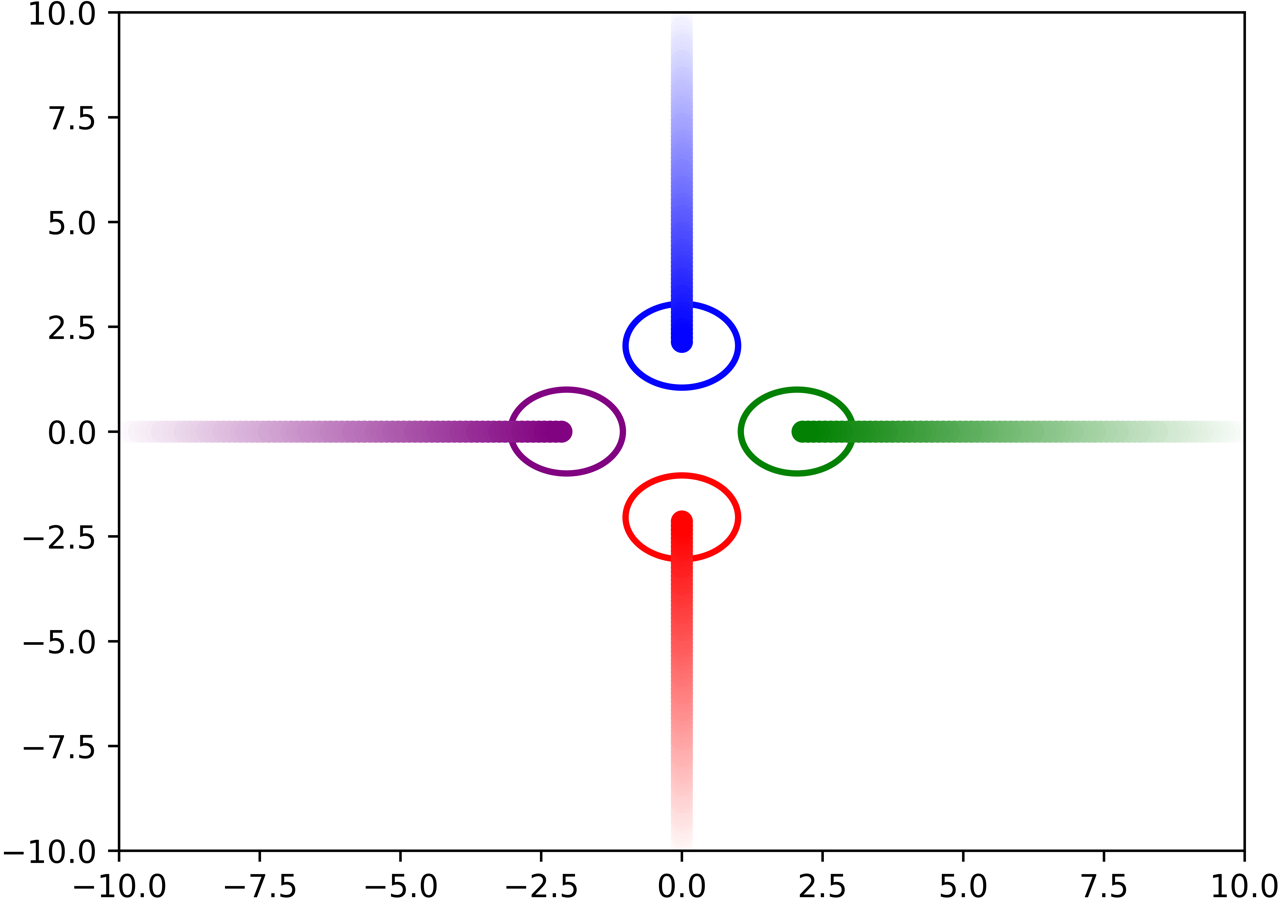}
 \caption{}
  \label{fig:deadlock}
\end{subfigure}
\begin{subfigure}{0.45\linewidth}
  \includegraphics[width = \linewidth]{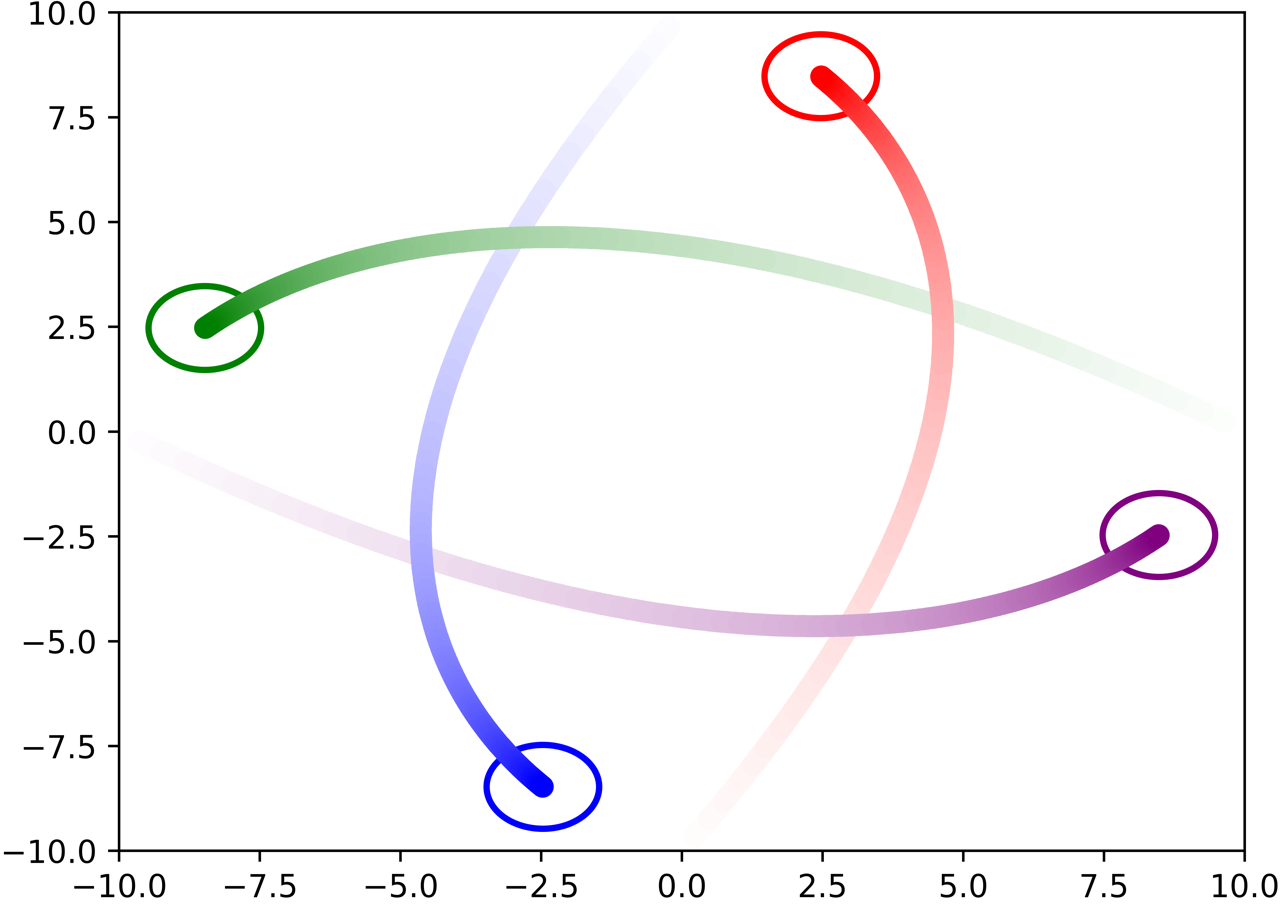}
 \caption{}
  \label{fig:baseline}
   \end{subfigure}
\caption{The multi-agent position swapping game using four single-integrator agents with the CLF-CBF method (Fig.~\ref{fig:deadlock}) and CLF-CBF with deadlock resolution method (Fig.~\ref{fig:baseline}).}
\label{fig:base}
\vspace{-0.7 cm}
\end{figure}

\begin{figure}[!htbp]
\centering
\begin{subfigure}{0.3\linewidth}
  \includegraphics[width = \linewidth]{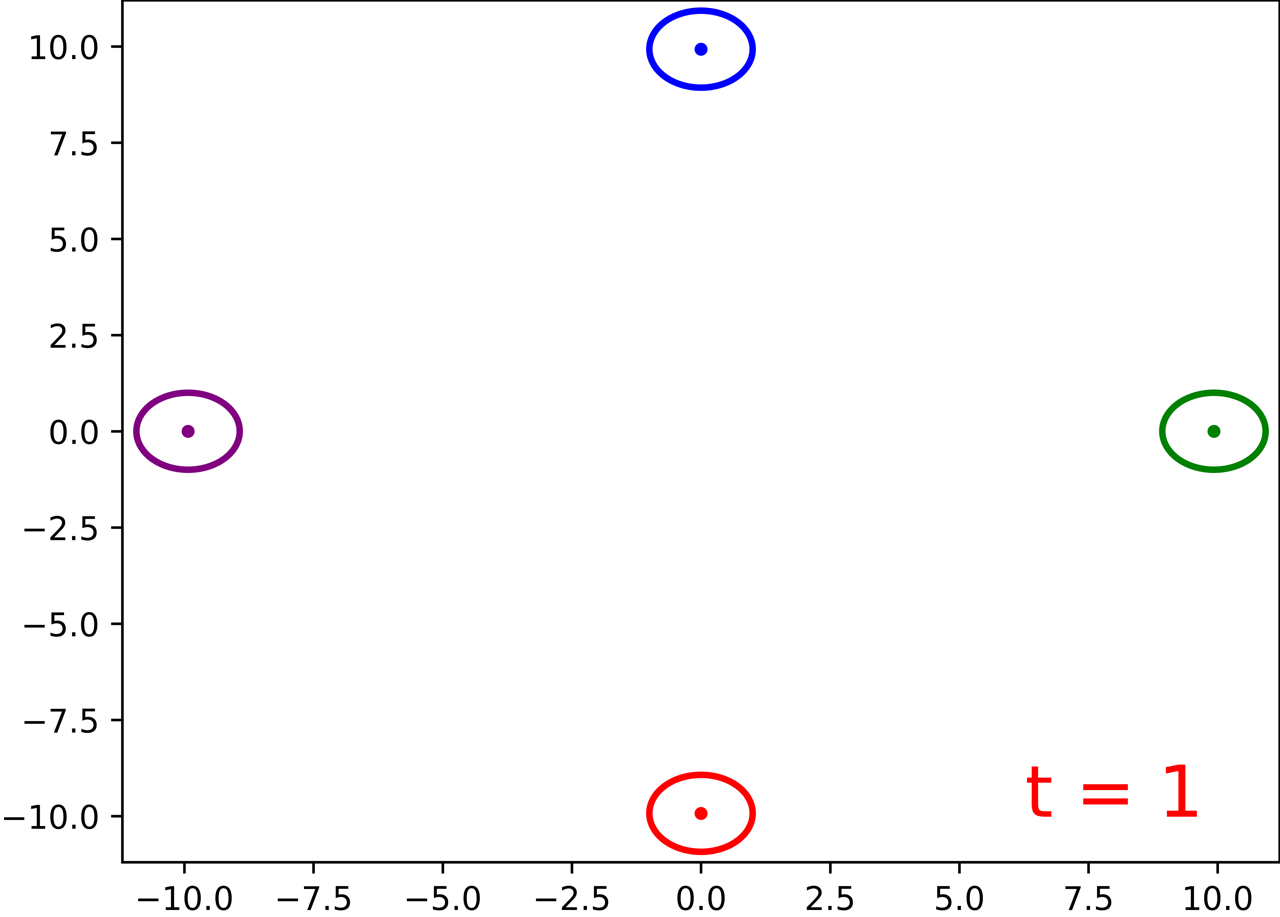}
\end{subfigure}
\begin{subfigure}{0.3\linewidth}
  \includegraphics[width = \linewidth]{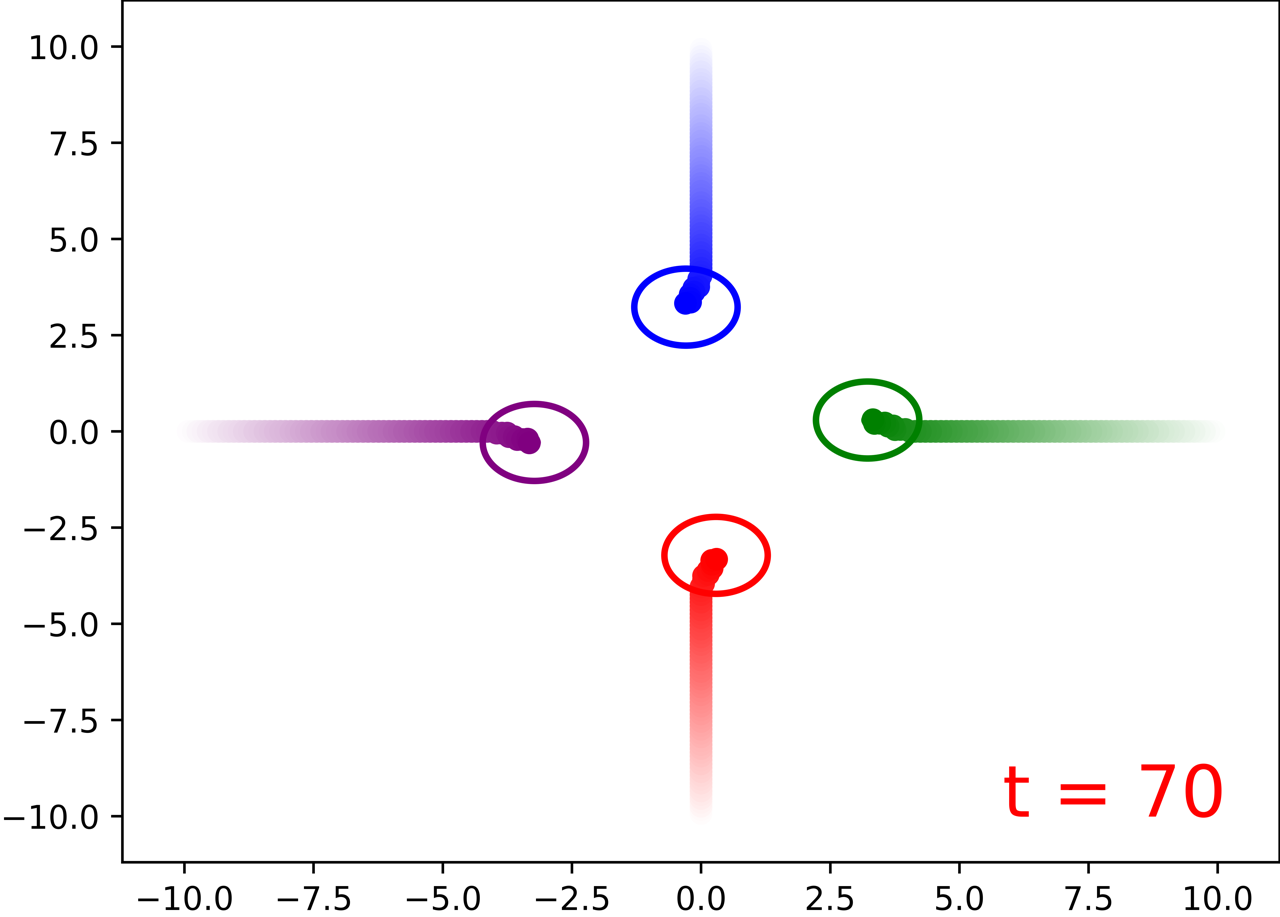}
   \end{subfigure}
\begin{subfigure}{0.3\linewidth}
  \includegraphics[width = \linewidth]{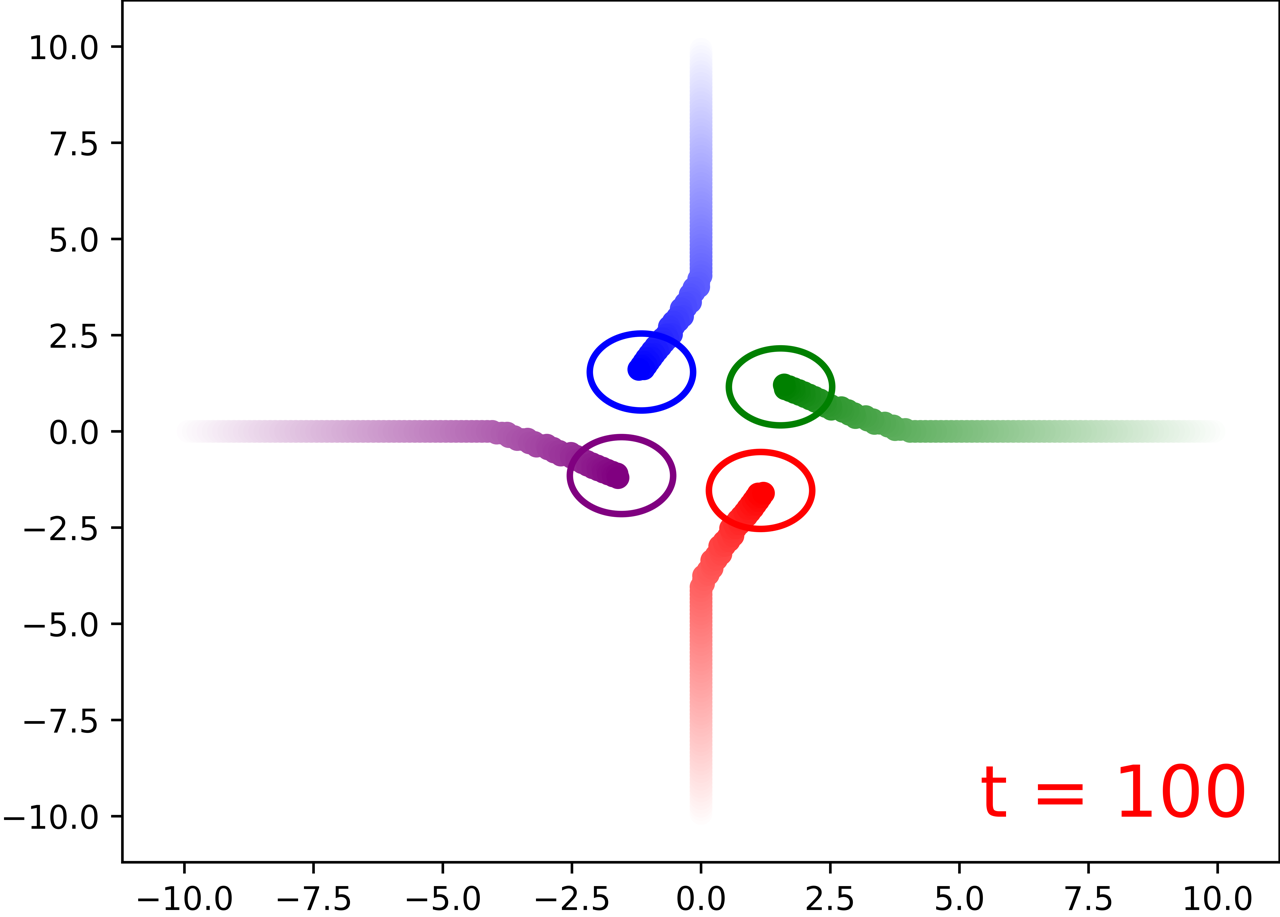}
   \end{subfigure}
\begin{subfigure}{0.3\linewidth}
  \includegraphics[width = \linewidth]{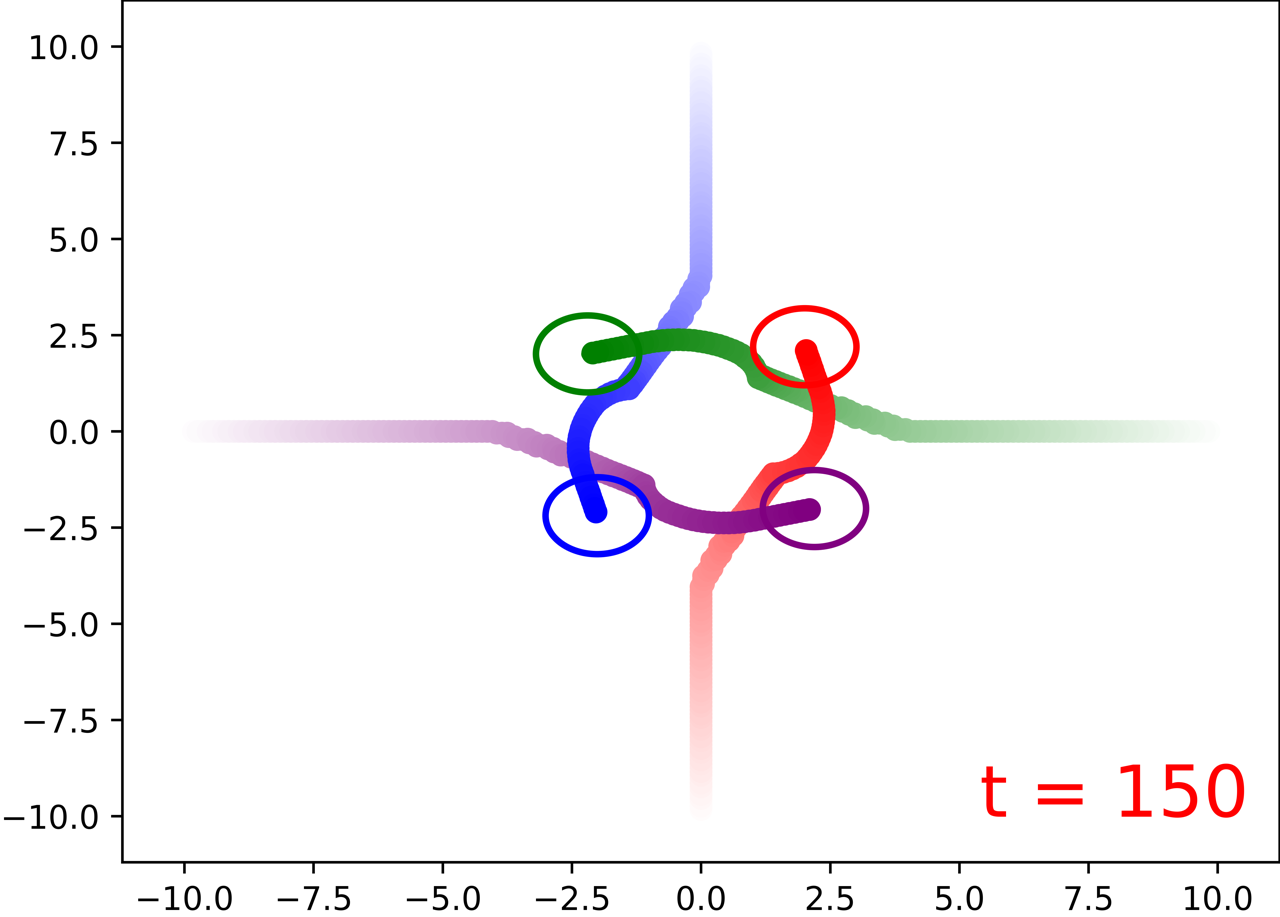}
   \end{subfigure}
\begin{subfigure}{0.3\linewidth}
  \includegraphics[width = \linewidth]{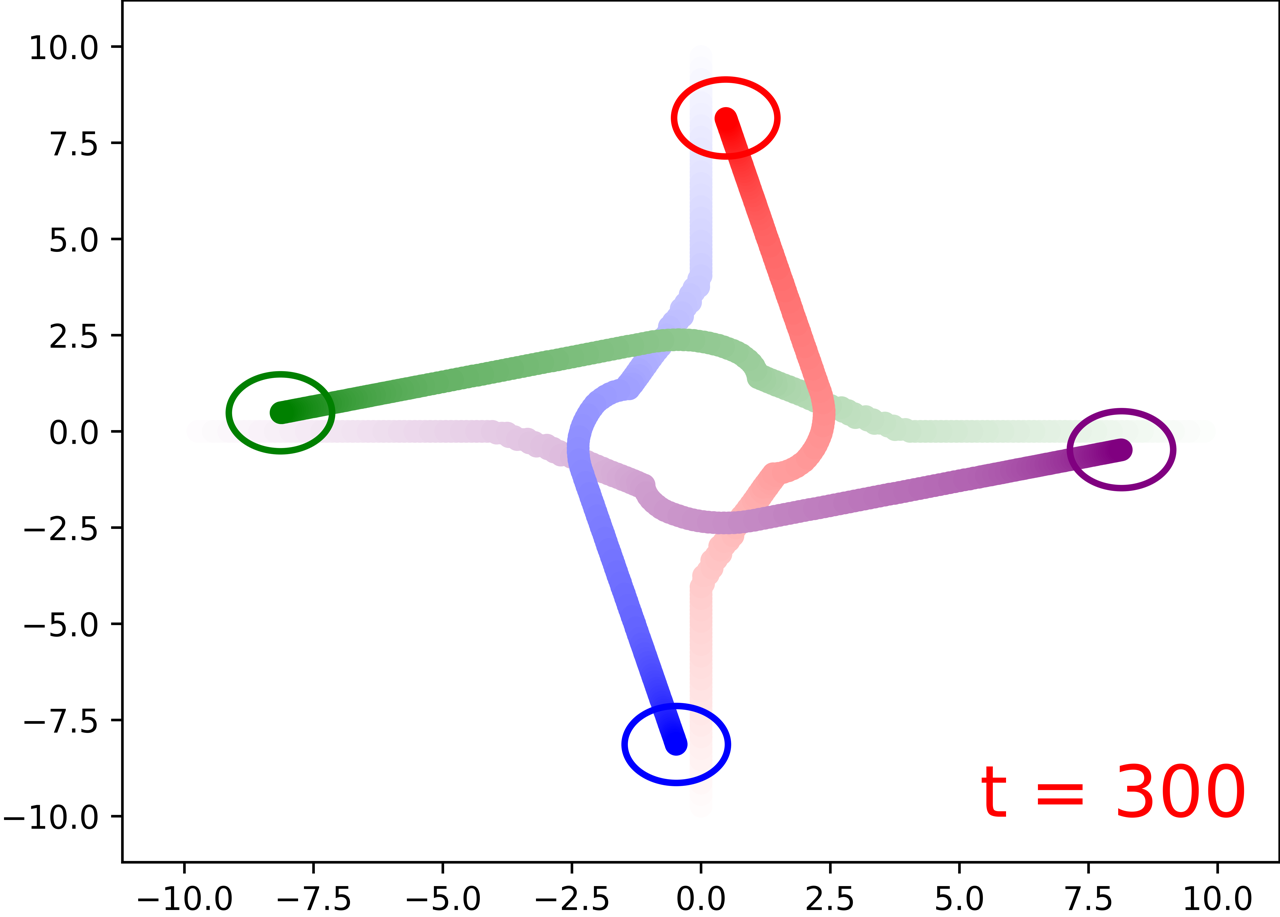}
   \end{subfigure}
\begin{subfigure}{0.3\linewidth}
  \includegraphics[width = \linewidth]{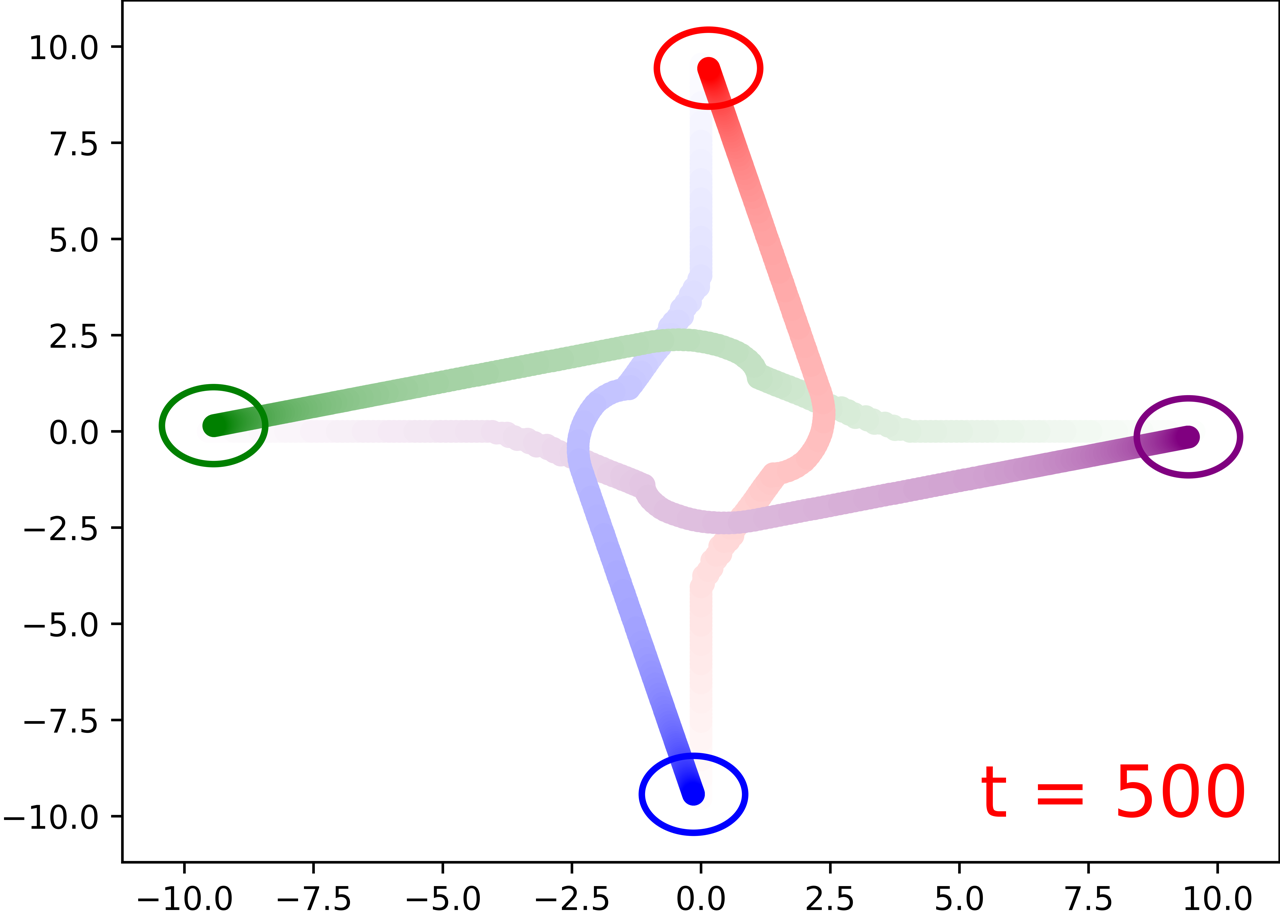}
   \end{subfigure}
\caption{The multi-agent position swapping game using four single-integrator agents with our method.}
\label{fig:prop}
\vspace{-0.7 cm}
\end{figure}

\subsection{Quantitative Results}
To verify the effectiveness of our method with different numbers of agents, we conducted experiments with up to 12 agents, ending at a fixed time step (i.e., $t = 500$). The agents are generated symmetrically across the environment. 
Fig.~\ref{fig:quanti} illustrates the superior performance of our algorithm in minimizing the distance to the target within a limited number of time steps. It clearly shows that our algorithm enables better task efficiency compared to the method in \cite{reis2020control}.

\begin{figure}[!htbp]
    \centering
    \includegraphics[width = 0.55\linewidth]{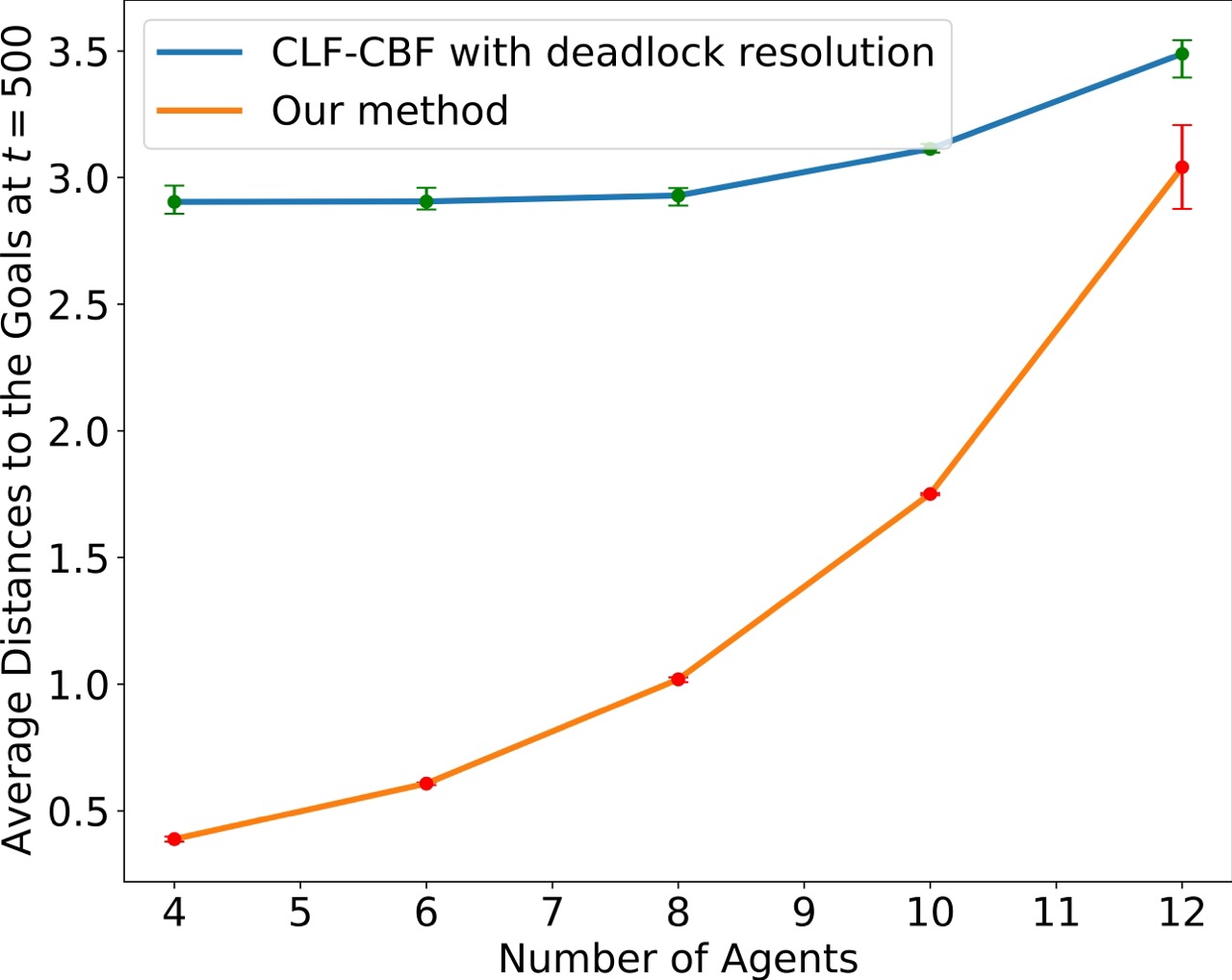}
    \caption{ Average distance to the target at $t = 500$, indicating the overall task efficiency. The figure is based on various scenarios, each involving a different number of agents starting from different initial positions.
    For each scenario, ten trials are conducted. 
    The error bars represent the min and max values.}
    \label{fig:quanti}
    \vspace{-0.4cm}
\end{figure}

\subsection{Real World Experiment}
We further use four Husarion ROSbot 2 PRO \cite{WinNT} robots as agents to play the multi-agent position-swapping game to verify the performance of our method. Besides, we put another robot at the center of the scenario as a static obstacle to verify the algorithm's effectiveness.
Several markers are set on the robot's body for the motion capture system OptiTrack acquiring and broadcasting the real-time robots' positions. We use the Robot Operating System(ROS) with a local internet with a router to establish connections between the OptiTrack and all four robots. After acquiring their positions, each robot uses the onboard computation resource to calculate the optimal control input for themselves. As shown in Fig.~\ref{fig:resolve}, the robots could automatically detour at the necessary positions. For the detailed implementation, we need to map the single-integrator dynamics to unicycle control commands for 
controlling the Husarion ROSbot 2PRO robots. Similarly to Robotarium \cite{wilson2020robotarium}, we use a near-identity diffeomorphism (NID) between the single-integrator and unicycle models \cite{olfati2002near} to enable the mapping from the single-integrator control to the unicycle control. Denote the robot state in the global frame as $[x, y, \theta]^{T}$. Then based on the optimized control $\mathbf{u}_i$ for single integrator dynamics using Eq.~\eqref{eq:prop}, we can get the control $\mathbf{v}_i$ for the unicycle dynamics:

{\footnotesize\begin{equation}
\mathbf{v}_i = 
\begin{bmatrix}
cos(\theta) & sin(\theta) \\
-\frac{1}{l}sin(\theta) & \frac{1}{l}cos(\theta)
\end{bmatrix} \mathbf{u}_i
\end{equation}}
$l$ is a small distance defining the mapped distance in alignment with the positive direction of the robot's orientation.

\begin{figure}
    \centering
    \includegraphics[width = 0.47\linewidth]{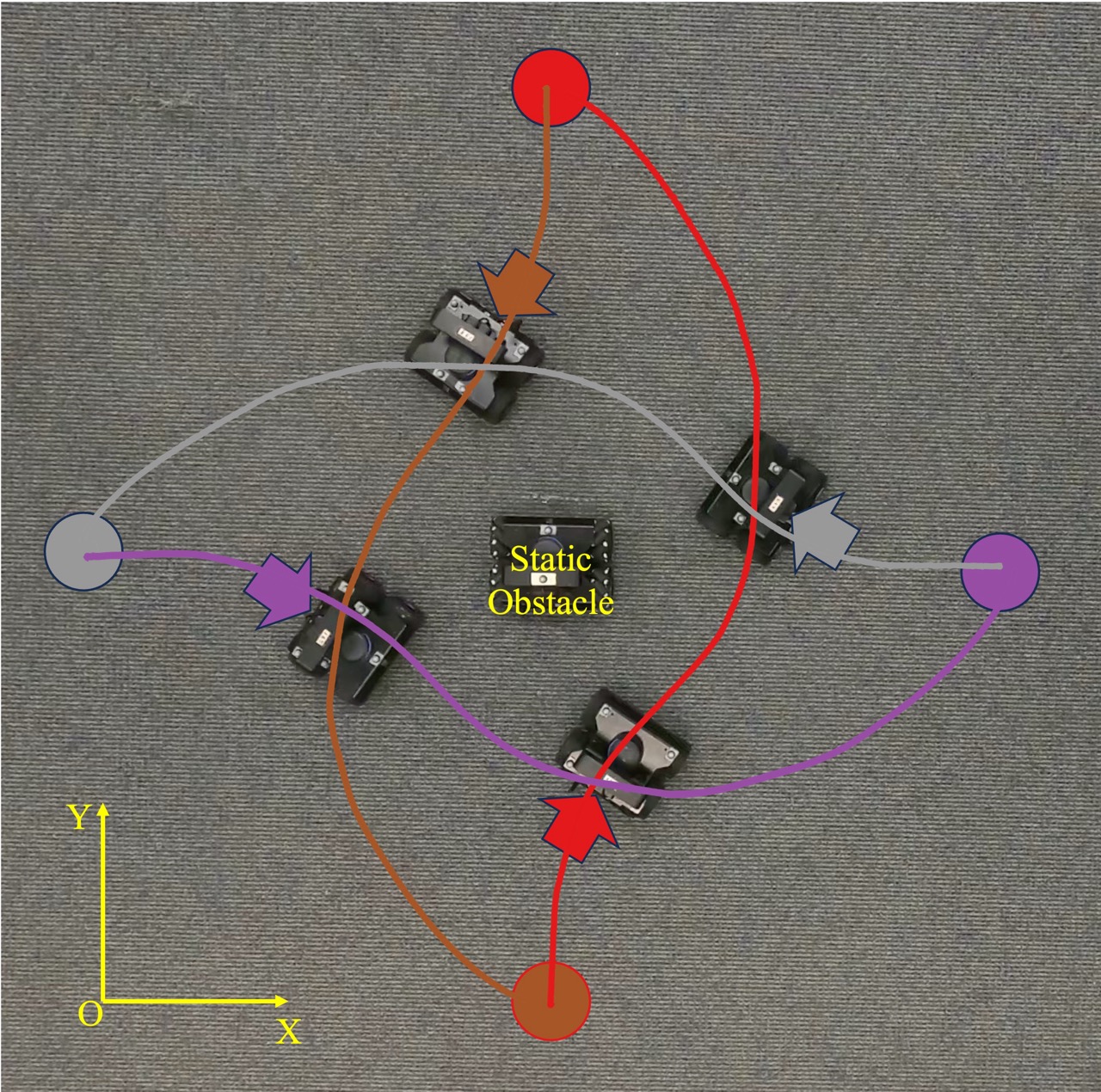}
    \caption{Four Husarion ROSbot 2 PRO robots \cite{WinNT} position swapping game with one static obstacle using our method. The filled circle with the same colored arrow represents the goal position of the corresponding robot. The curves are the corresponding real trajectories.}
    \label{fig:resolve}
    \vspace{-0.8cm}
\end{figure}

\section{Conclusion}
In this paper, we present a generalized control framework for the decentralized multi-agent system, which enables efficient task execution and deadlock resolution. 
In particular, we proposed a deadlock indicator function to determine when deadlock resolution is needed, and synthesized the control objectives of task-prescribed stability, collision avoidance and deadlock resolution to the unified framework. This allows individual agents to adaptively switch on/off the deadlock resolution when necessary while moving towards the goals with collision-free trajectories. The resulting decentralized multi-agent controller is shown to be valid and effective in numerical simulation and real-world experiments.

\bibliographystyle{IEEEtran}
\bibliography{ref}

\end{document}